%% file: 0main.tex
\documentclass[runningheads,envcountsame]{llncs}
\usepackage[T1]{fontenc}
\usepackage{graphicx} 
\usepackage{lineno}
\usepackage{url}
\title{Checking History-Determinism is NP-hard \newline for Parity Automata}
\author{Aditya Prakash\inst{1}\orcidlink{0000-0002-2404-0707}}
\authorrunning{Aditya Prakash}
\institute{University of Warwick, UK \newline
\email{aditya.prakash@warwick.ac.uk}}
\input{00macros} 
\begin{document}
\maketitle
\begin{abstract}
 We show that the problem of checking if a given nondeterministic parity automaton simulates another given nondeterministic parity automaton is NP-hard. We then adapt the techniques used for this result to show that the problem of checking history-determinism for a given parity automaton is NP-hard. This is an improvement from Kuperberg and Skrzypczak's previous lower bound of solving parity games from 2015. We also show that deciding if Eve wins the one-token game or the two-token game of a given parity automaton is NP-hard. Finally, we show that the problem of deciding if the language of a nondeterministic parity automaton is contained in the language of a history-deterministic parity automaton can be solved in quasi-polynomial time.
\end{abstract}

\section{Introduction}
\input{1intro}
\section{Preliminaries}
\input{2prelim}
\section{Simulation is $\NP$-hard}\label{sec:NP-hardnessforsimulation}
\input{3simulation}
\section{Checking History-Determinism is $\NP$-hard}\label{sec:NP-hardnessforHD}
\input{4HDisNPhard}
\section{Language Containment}\label{sec:lang-inclusion}
\input{5Inclusion}
\section{Discussion}
\input{6conclusion}
\bibliographystyle{splncs04}
\bibliography{gfg} 
\newpage
\appendix
\section{Appendix for \cref{sec:NP-hardnessforHD}}
\input{appendix1}
\section{Appendix for \cref{sec:lang-inclusion}}
\input{appendix2}
\end{document}

%% file: 00macros.tex
\usepackage{amsfonts}
\usepackage{amsmath}
\usepackage{bm}
\usepackage{amssymb}
\usepackage[pdftex,dvipsnames]{xcolor}
\usepackage{tikzlings}
\usepackage{orcidlink}
\usepackage{tikzducks}
\usepackage{tikzsymbols}
\usepackage{tikz}
\usepackage{enumitem}
\usepackage{mathtools}
\usepackage{thmtools}
\usepackage{thm-restate}
\usepackage{hyperref}
\usepackage[capitalise]{cleveref}
\usepackage{xspace}
\usepackage[new]{old-arrows}
\usepackage{float}
\usepackage{thm-restate}
\usepackage[bold,full]{complexity}

\newcommand{\hdcontainment}{{\scshape HD-automaton containment}\xspace}

\newcommand{\hdproblem}{{\scshape History-deterministic}\xspace}
\newcommand{\simproblem}{{\scshape Simulation}\xspace}
\newcommand{\twodimparity}{{\scshape 2-D parity game}\xspace}
\newcommand{\specialparity}{{\scshape Good 2-D parity game}\xspace}
\newcommand{\eve}{\pmb{\exists}}
\newcommand{\adam}{\pmb{\forall}}
\newcommand{\Gc}{\mathcal{G}} 
\newcommand{\Lc}{\mathcal{L}}
\newcommand{\Ac}{\mathcal{A}}
\newcommand{\Fc}{\mathcal{F}}
\newcommand{\Hc}{\mathcal{H}}

\newcommand{\Bc}{\mathcal{B}}
\newcommand{\Dc}{\mathcal{D}}

\newcommand{\Nc}{\mathcal{N}}

\newcommand{\Oc}{\mathcal{O}}
\newcommand{\Pc}{\mathcal{P}}

\newcommand{\honest}{\text{uncorrupted}\xspace}
\newcommand{\dishonest}{\text{corrupted}\xspace}
\newcommand{\speshial}{good\xspace}
\newclass{\TOWER}{TOWER}
\newclass{\ACKERMANN}{ACKERMANN}
\newclass{\EXPTIME}{EXPTIME}
\newclass{\IIEXPTIME}{2-EXPTIME}
\usepackage{graphicx}
\usepackage{array}
\usepackage{tikz}
\usetikzlibrary{automata, arrows, positioning, decorations.markings, decorations.pathreplacing}
\usetikzlibrary {decorations.pathmorphing, decorations.shapes}
\usetikzlibrary{backgrounds,fit,shapes.geometric}
\usetikzlibrary {graphs, quotes}
\usetikzlibrary{calc}
\usetikzlibrary{shapes,shapes.geometric,fit}
\usepackage[new]{old-arrows}
\usetikzlibrary {graphs,shapes.geometric,quotes}
\usetikzlibrary{decorations.pathreplacing}
\usetikzlibrary{automata}
\usetikzlibrary{calc}
\usetikzlibrary{arrows,automata,decorations.pathmorphing}
\usetikzlibrary{shapes,shapes.geometric,fit,calc,automata}
\usepackage{xargs} 

%% file: 1intro.tex
Deciding language inclusion between two automata is a fundamental problem in verification, wherein we ask whether all executions of an implementation satisfy a given specification.  Unfortunately, the problem of checking language inclusion is often computationally hard. For parity automata---which are the focus of this paper---it is $\PSPACE$-complete, with $\PSPACE$-hardness already occurring for finite state automata~\cite{SM73}. 

On the other hand, simulation is a fundamental behavioural relation between two automata~\cite{Mil71,HKR02}, which is a finer relation than language inclusion and is easier to check. For parity automata, simulation can be decided in polynomial time if the parity indices are fixed; otherwise it is in $\NP$~\cite{CHP07}. Note that while simulation between two automata is sufficient to guarantee language inclusion, it is not necessary. 

For history-deterministic automata, however, the relation of language inclusion is equivalent to simulation~\cite{BHLST23,BL23}, thus making them suitable for verification. These are nondeterministic automata where the nondeterminism can be resolved `on-the-fly', just based on the prefix of the word read so far. The definition we use here was introduced by Henzinger and Piterman in 2006, where they dubbed it `good-for-games' automata, while the term `history-determinism' was coined by Colcombet~\cite{Col09} in the context of regular cost automata. 

History-deterministic parity automata are more succinct than their deterministic counterparts~\cite{KS15} whilst still maintaining tractability for the problems of verification and synthesis on them~\cite{HP06,KS15,BL23}. Consequently, history-deterministic parity automata have been the subject of extensive research~\cite{KS15,BKS17,BK18,Sch20,RK22,Kup22}, and has garnered significant attention over the recent years beyond parity automata as well, extending to quantitative automata~\cite{BL21,BL22}, infinite state systems~\cite{GJLZ21,LZ22,PT23,BPT23,EGJLZ23},  and timed automata~\cite{BHLST23}. 

Despite these recent research efforts, a significant gap remains in understanding the complexity of checking whether a given parity automaton is history-deterministic. While Henzinger and Piterman have shown an $\EXPTIME$ upper bound~\cite{HP06}, the best lower bound known so far is by Kuperberg and Skrzypczak since 2015~\cite{KS15}, who showed that checking for history-determinism is at least as hard as finding the winner of a parity game~\cite{KS15}---a problem that can be solved in quasi-polynomial time and is in $\NP \cap \coNP$ (and even in $\UP \cap \coUP$~\cite{Jur98}). 

Kuperberg and Skrzypczak also gave a polynomial-time algorithm to check for history-determinism of co-B\"uchi automata in their work~\cite{KS15}. 
This was followed by a polynomial time algorithm to check for history-determinism of B\"uchi automata in 2018 by Bagnol and Kuperberg~\cite{BK18}, 
who showed that in order to check if a B\"uchi automaton is history-deterministic, it suffices to find the winner of the so-called `two-token game' of the automaton. 
This connection between history-determinism and two-token games was extended in 2020 to co-B\"uchi automata by Boker, Kuperberg, Lehtinen, and Skrzypczak~\cite{BKLS20b}. 
It is conjectured that the winner of the two-token game of a parity automaton characterises its history-determinism. While the two-token conjecture is open to date, showing this conjecture would imply that one can check history-determinism of a given parity automata with a fixed parity index in polynomial time.

\subsubsection*{Our contributions.}
We show that checking for simulation between two parity automata is $\NP$-hard when the parity index is not fixed. Since simulation is known to be in $\NP$, this establishes the problem to be $\NP$-complete (\cref{thm:Sim-is-NP-hard}). 

An adaptation of our proof of \cref{thm:Sim-is-NP-hard} gives us that checking history-determinism for a parity automata is also $\NP$-hard (\cref{thm:HD-is-NP-hard}), when the parity index is not fixed. This is an improvement on Kuperberg and Skrzypczak's result from 2015, which shows that checking history-determinism for parity automata is at least as hard as solving parity games~\cite{KS15}. We also show, using the same reduction, that checking whether Eve wins the 2-token game (of a given parity automaton) is $\NP$-hard, while checking whether Eve wins the 1-token game is $\NP$-complete (\cref{thm:HD-is-NP-hard}).

As remarked earlier, for history-deterministic parity automata, the relation of language inclusion is equivalent to simulation. This gives us an immediate $\NP$ upper bound for checking language inclusion of a nondeterministic parity automaton in an HD-parity automata, as was observed by Schewe~\cite{Sch20}. We show that we can do better, by showing the problem to be decidable in quasi-polynomial time (\cref{thm:inclusion-theorem}). 

\subsubsection*{Overview of the paper: one reduction for all.}
The central problem used in our reduction is of checking whether Eve wins a 2-D parity game, which is known to be $\NP$-complete due to Chatterjee, Henzinger and Piterman~\cite{CHP07}. In \cref{sec:NP-hardnessforsimulation}, we give a reduction from this problem to checking for simulation between two parity automata, thus establishing its $\NP$-hardness (\cref{thm:Sim-is-NP-hard}). We then show, in \cref{sec:good-games}, that the problem of checking whether Eve wins a \speshial 2-D parity games---a technical subclass of 2-D parity games---is also $\NP$-hard. In \cref{sec:good-to-hd}, we show that modifying the reduction in proof of  \cref{thm:Sim-is-NP-hard} to take as inputs \speshial 2-D parity games yields $\NP$-hardness for the problems of checking history-determinism (\cref{lem:hd-is-np-hard}) and of checking if Eve wins the 1-token game or the 2-token game (\cref{thm:HD-is-NP-hard}). Finally, in \cref{sec:lang-inclusion}, we give a quasi-polynomial algorithm to check whether the language of a nondeterministic parity automaton is contained in the language of a history-deterministic parity automaton (\cref{thm:inclusion-theorem}), by reducing to finding the winner in a parity game. 

%% file: 2prelim.tex
We let $\mathbb{N}=\{0,1,2,\cdots \}$ to be the set of natural numbers, and $\omega$ to be the cardinality of $\mathbb{N}$. We will use $[i,j]$ to denote the set of integers in the interval $\{i,i+1,\ldots,j\}$ for two natural numbers $i,j$ with $i<j$, and $[j]$ for the interval $[0,j]$. An \textit{alphabet} $\Sigma$ is a finite set of \textit{letters}. We use $\Sigma^{*}$ and $\Sigma^{\omega}$ to denote the set of words with finite and $\omega$ length over $\Sigma$ respectively. We also let $\varepsilon$ denote the unique word of length 0.

\subsection{Parity conditions}
Let $G = (V,E)$ be a (finite or infinite) directed graph equipped with a \emph{priority function} $\chi:E \xrightarrow{} \mathbb{N}$ that assigns each edge with a natural number, called its \emph{priority}. We say that an infinite path $\rho$ in $G$ satisfies the \emph{$\chi$-parity condition} if the highest priority occurring infinitely often in the path is even. When clear from the context, we will drop `parity condition' and instead  say that $\rho$ satisfies $\chi$. 

A parity condition is easily \emph{dualised}. Given a priority function $\chi$ as above, consider the priority function $\chi' := \chi +1$ that is obtained by increasing all the labels by 1. Then, an infinite path satisfies $\chi'$ if and only if it does not satisfy $\chi$.

\subsection{Parity automata}
A \emph{nondeterministic parity automaton} $\Ac = (Q,\Sigma,q_0,\Delta,\Omega)$ contains a finite directed graph with edges labelled by letters in $\Sigma$. These edges are called \emph{transitions}, which are elements of the set $\Delta \subseteq Q \times \Sigma \times Q$, and the vertices of this graph are called \emph{states},  which are elements of the set $Q$. 

Each automaton has a designated \emph{initial state} $q_0 \in Q$, and a priority function $\Omega: \Delta \xrightarrow{} [i,j]$ which assigns each transition a \emph{priority} in $[i,j]$, for $i<j$ two natural numbers. For states $p,q$ and an alphabet $a \in \Sigma$, we use $p\xrightarrow{a:c}q$ to denote a transition from $p$ to $q$ on the letter $a$ that has the priority $c$. 

A \emph{run} on an infinite word $w$ in $\Sigma^{\omega}$ is an infinite path in the automaton, starting at the initial state and following transitions that correspond to the letters of $w$ in sequence.   We say that such a run is \emph{accepting} if it satisfies the $\Omega$-parity condition, and a word $w$ in $\Sigma^{\omega}$ is accepting if the automaton has an accepting run on $w$. The \emph{language} of an automaton $\Ac$, denoted by $L(\Ac)$, is the set of words that it accepts. We say that the automaton $\Ac$ \emph{recognises} a language $\Lc$ if $L(\Ac)=\Lc$. A parity automaton $\Ac$ is said to be \emph{deterministic} if for any given state in $\Ac$ and any given letter in $\Sigma$, there is at most one transition from the given state on the given letter. 

If $\Ac$'s priorities are in $[i,j]$, we say that $(j-i+1)$ is the \emph{number of priorities} of $\Ac$. Since decreasing (or increasing) all of these priorities in the automaton by 2 does not change the acceptance of a run--- and hence a word---in the automaton, we will often assume $i$ to be $0$ or $1$. With this assumption, the interval $[i,j]$ is then said to be the \emph{parity index} of $\Ac$. A \emph{B\"uchi (resp. co-B\"uchi)} automaton is a parity automaton whose parity index is $[1,2]$ (resp. $[0,1]$).

\begin{remark}
    We note that we allow an automatonto be \emph{incomplete}, i.e. there might be letter and state pairs in an automaton such that there are no transitions on that letter from that state.
\end{remark}
\subsection{Game arenas}
An \emph{arena} is a directed graph $G=(V,E)$ with vertices partitioned as $V_{\adam}$ and $V_{\eve}$ between two players Adam and Eve respectively. Additionally, a vertex $v_0 \in V_{\adam}$ is designated as the initial vertex. We say that the set of vertices $V_{\eve}$ is owned by Eve while the set of vertices $V_{\adam}$ is owned by Adam. Additionally, we assume that the edges $E$ don't have both its start and end vertex in $V_{\eve}$ or $V_{\adam}$.

Given an arena as above, a \emph{play} of this arena is an infinite path starting at $v_0$, and is formed as follows. A play starts with a token at the start vertex $v_0$, and proceeds for countably infinite rounds. At each round, the player who owns the vertex on which the token is currently placed chooses an outgoing edge, and the token is moved along this edge to the next vertex for another round of play. This creates an infinite path in the arena, which we call a play of $G$. 

A \emph{game} $\Gc$ consists of an arena $G=(V,E)$ and a winning condition given by a language $L \subseteq E^{\omega}$. We say that Eve \emph{wins a play} $\rho$ in $G$ if $\rho$ is in $L$, and Adam wins otherwise. A \emph{strategy} for Eve in such a game $\Gc$ is a function from the set of plays that end at an Eve's vertex to an outgoing edge from that vertex. Such an Eve strategy is said to be a \emph{winning strategy} for Eve if any play that can be produced when she plays according to her strategy is winning for Eve. We say that Eve \emph{wins the game} if she has a winning strategy. Winning strategies are defined for Adam analogously, and we say that Adam wins the game if he has a winning strategy.  

In this paper we will deal with \emph{$\omega$-regular games}. These are games where the languages specifying the winning condition are  recognised by a parity automata. Such games are known to be determined~\cite{Mar75,GH82}, i.e. each game has a winner. Two games are \emph{equivalent} if they have the same winner.

\subsection{Parity games}
A \emph{parity game} $\Gc$ is played over a finite game arena $G = (V,E)$, with the edges of $G$ labelled by a priority function $\chi:E \xrightarrow{} \{0,1,2,\cdots,d\}$. A play $\rho$ in the arena of $\Gc$ is winning for Eve if and only if $\rho$ satisfies the $\chi$-parity condition. 

\subsection{Muller conditions and Zielonka trees}
A $(C,\Fc)$-Muller conditions consists of a finite set of colours $C$, and a set $\Fc$ consisting of subsets of $C$. An infinite sequence in $C^{\omega}$ satisfies the $(C,\Fc)$-Muller condition if the set of colours seen infinitely often along the sequence is in $\Fc$.

A \emph{Muller game} $\Gc$ consists of an arena $G = (V,E)$, a colouring function $\pi: E \xrightarrow{} C$ and a Muller condition $(C,\Fc)$. An infinite play $\rho$  in $\Gc$ is winning for Eve if the set of colours seen infinitely often along the play is in $\Fc$, and Eve wins the Muller game $\Gc$ if she has a winning strategy.

Every Muller game can be converted to an equivalent parity game, as shown by Gurevich and Harrington~\cite{GH82}. We will use the conversion of Dziembowski, Jurdzi\'nski, and Walukiewikz that involve Zielonka trees~\cite{DJW97,CCF21}, which we define below. 
\begin{definition}[Zielonka tree]\label{def:Zielonkatree}
    Given a Muller condition $(C,\Fc)$, the \emph{Zielonka tree} of a Muller condition, denoted $Z_{C,\Fc}$, is a tree whose nodes are labelled by subsets of $C$, and is defined inductively. The root of the tree is labelled by $C$. For a node that is already constructed and labelled with the set $X$, its children are nodes labelled by distinct maximal non-empty subsets $X' \subsetneq X$ such that $X \in \Fc \Leftrightarrow X' \notin \Fc$.  If there are no such $X'$, then the node labelled $X$ is a leaf of $Z_{C,\Fc}$ and has no attached children.
\end{definition}

Given a $(C,\Fc)$-Muller condition, consider the language $L \subseteq C^{\omega}$ consisting of words $w$ that satisfy the $(C,\Fc)$-Muller condition. The language $L$ is then said to be the language of the $(C,\Fc)$-Muller condition, and can be recognised by a deterministic parity automaton, whose size depends on the size of the Zielonka tree~\cite{CCF21}. 
\begin{lemma}[\cite{CCF21}]\label{lem:zielonka-automata}
    Let $(C,\Fc)$ be a Muller condition with the Zielonka tree $Z_{C,\Fc}$ that has $n$ leaves and height $h$. Then there is a deterministic parity automaton $\Dc_{C,\Fc}$ that can be constructed in polynomial time such that $\Dc_{C,\Fc}$ has $n$ states and  $(h+1)$ priorities, and accepts the language of the the $(C,\Fc)$-Muller condition.  
\end{lemma}

Consider a Muller game $\Gc$ on the arena $G = (V,E)$ with the colouring function $\pi:E \xrightarrow{} C$ and the Muller condition $(C,\Fc)$. We can then construct an equivalent parity game $\Gc'$ by taking the product of $\Gc$ with the automaton $\Dc_{C,\Fc}$ from \cref{lem:zielonka-automata}. In more details, the set of vertices $V'$ of $\Gc'$ consists of vertices of the form $v'=(v,q)$, where $v$ is a vertex in $\Gc$ and $q$ is a state in $\Dc_{C,\Fc}$. The owner of the vertex $(v,q)$ is the owner of the vertex $v$, and the initial vertex is $(\iota,q_0)$, where $\iota$ is the initial vertex in $\Gc$ and $q_0$ is the initial state in $\Dc_{C,\Fc}$. We have the edge $e' = (v,q) \xrightarrow{} (v',q')$ in $\Gc'$ if $e = v \xrightarrow{} v'$ is an edge in $\Gc$ with the colour $\pi(e) = c$, and $\delta = q \xrightarrow{c} q'$ is a transition in $\Dc_{C,\Fc}$. The edge $e'$ is assigned the priority $\Omega(\delta)$ in $\Gc'$, where $\Omega$ is the priority function of the automaton $\Dc_{C,\Fc}$. The game $\Gc'$ then is such that Eve wins $\Gc$ if and only if Eve wins $\Gc'$.

\begin{lemma}\label{lem:Muller to Parity reduction}
    Let $\Gc$ be a Muller game on an arena consisting of $m$ vertices with a Muller condition $(C,\Fc)$ whose Zielonka tree $Z_{C,\Fc}$ has $n$ leaves and height $h$. Then, $\Gc$ can be converted to an equivalent parity game $\Gc'$ which has  $mn$ many vertices and $h+1$ priorities. 
\end{lemma}

\subsection{2-dimensional parity game}\label{def:twodimparitygame}
Multi-dimensional parity games were introduced by Chatterjee, Henzinger and Piterman, where they called it generalised parity games~\cite{CHP07}. For our purposes, it suffices to consider 2-dimensional (2-D) parity games, which is what we define now. 

A \emph{2-dimensional parity game} $\Gc$ is similar to a parity game, but we now have two priority functions $\pi_1:E \xrightarrow{} [0,d_1]$ and $\pi_2:E \xrightarrow{} [0,d_2]$ on $E$. Any infinite play in the game is winning for Eve if the following holds:
\emph{if the play satisfies $\pi_1$, then it satisfies $\pi_2$}.\\ 
We say that Adam wins the game otherwise. We call the problem of deciding whether Eve wins a 2-D parity game as {\scshape 2-D parity game}. 
\begin{quote}
    {\scshape 2-D parity game}: Given a 2-D parity game $\Gc$, does Eve win $\Gc$?
\end{quote}
If Eve has a strategy to win a 2-D parity game, then Eve has a positional winning strategy to do so, i.e. she can win by always choosing the same edge from each vertex in $V_{\eve}$, which is given by a function $\sigma:V_{\eve} \xrightarrow{} E$. This can be inferred directly from seeing the 2-D parity game as a Rabin game, which are known to have positional strategies for Eve~\cite{Emerson85}. Furthermore, given a positional strategy $\sigma$ for Eve in a 2-D parity game (or a Rabin game), one can check in polynomial time if $\sigma$ is a winning strategy~\cite{Emerson85}. This gives us a nondeterministic polynomial time procedure to decide if Eve wins a given 2-D parity game. In 1988, Emerson and Jutla established $\NP$-hardness for Rabin games~\cite{EJ88,EJ99}. This was later extended by Chatterjee, Henzinger, and Piterman in 2007 to show $\NP$-hardness for 2-D parity games as well~\cite{CHP07}.

\begin{theorem}[\cite{CHP07}]    
   The problem of deciding whether Eve wins a given 2-D parity game is $\NP$-complete.
\end{theorem}

\begin{remark}
Chatterjee, Henzinger and Piterman give a slightly different and a more natural definition of 2-D parity games~\cite{CHP07}, where the winning condition for Eve requires every play to satisfy either of two given parity conditions. It is easy to see, however, that both definitions are log-space inter-reducible to each other, by dualising the first parity condition. Our definition, although less natural, makes the connection to simulation games and our reductions in \cref{sec:NP-hardnessforHD,sec:NP-hardnessforsimulation} more transparent. 
\end{remark}

\subsection{Simulation}\label{subsec:Prelims-simulation}
We say a parity automaton $\Ac$ \emph{simulates} another parity automaton $\Bc$ if for any (finite or infinite) run on $\Bc$, there is a corresponding run on $\Ac$ on the same word that can be constructed on-the-fly such that if the run in $\Bc$ is accepting, so is the corresponding run in $\Ac$. This is made more formal by the following \emph{simulation game}.

\begin{definition}[Simulation game]
    Given nondeterministic parity automata $\Ac = (Q,\Sigma,q_0,\Delta_A,\Omega_A)$ and $\Bc = (P,\Sigma,p_0,\Delta_B, \Omega_B)$, the \emph{simulation game} between $\Ac$ and $\Bc$, denoted $Sim(\Ac,\Bc)$, is defined as a two player game between Adam and Eve as follows, with positions in $P \times Q$. A play of the simulation game starts at the position $(p_0,q_0)$, and has $\omega$ many rounds. For each $i \in \mathbb{N}$, the $(i+1)^{th}$ round starts at a position $(p_i,q_i) \in P \times Q$, and proceeds as follows:
    \begin{itemize}
        \item Adam selects a letter $a \in \Sigma$, and a transition $p_i \xrightarrow{a} p_{i+1}$ in $\Bc$.
        \item Eve selects a transition $q_i \xrightarrow{a} q_{i+1}$ on the same letter in $\Ac$.
    \end{itemize}
   The new position is $(p_{i+1},q_{i+1})$, for another round of the play.

   The player Eve wins the above play if either her constructed run in $\Ac$ is accepting, or Adam's constructed run in $\Bc$ is rejecting. If Eve has a winning strategy in $Sim(\Ac,\Bc)$, then we say that $\Ac$ \emph{simulates} $\Bc$, and denote it by $\Bc \lesssim \Ac$.
\end{definition}
We call the problem of checking whether a parity automaton simulates another as {\scshape Simulation}:
\begin{quote}
    {\scshape Simulation:} Given two parity automata $\Ac$ and $\Bc$, does $\Ac$ simulate $\Bc$?
\end{quote}
The simulation game $Sim(\Ac,\Bc)$ can naturally be seen as a 2-D parity game, where the arena is the product of two automata with Adam selecting letters and transitions in $\Ac$ and Eve transitions in $\Bc$, and the priority functions $\chi_1$ and $\chi_2$ based on corresponding priorities of transitions in $\Ac$ and $\Bc$ respectively. Since \twodimparity can be solved in $\NP$, \simproblem can be solved in $\NP$ as well.

\subsection{History-determinism}

A \emph{history-deterministic} (HD) parity automaton is a nondeterministic parity automaton in which the nondeterminism can be resolved `on-the-fly' just based on the prefix read so far, without knowing the rest of the word. The history-determinism of a parity automaton can be characterised by the letter game, which is a 2-player turn-based game between Adam and Eve, who take alternating turns to select a letter and a transition in the automaton (on that letter), respectively. After the game ends, the sequence of Adam's choices of letters is an infinite word, and the sequence of Eve's choices of transitions is a run on that word. Eve wins the game if her run is accepting or Adam's word is rejecting, and we say that an automaton is history-deterministic if Eve has a winning strategy in the history-determinism game.

\begin{definition}[Letter game]
Given a parity automaton $\Ac = (Q,\Sigma, q_0, \Delta,\Omega)$, the \emph{letter game} of $\Ac$ is defined between the two players Adam and Eve as follows, with positions in $Q \times \Sigma^{*}$. The game starts at $(q_0,\varepsilon)$ and proceeds in $\omega$ many rounds. For each $i \in \mathbb{N}$, the $(i+1)^{th}$ round starts at a position $(q_i,w_i) \in Q \times \Sigma^{i}$, and proceeds as follows:\begin{itemize}
    \item Adam selects a letter $a_i \in \Sigma$
    \item Eve selects a transition $q_i \xrightarrow{a_i} q_{i+1} \in \Delta$
\end{itemize}
The new position is $(q_{i+1},w_{i+1})$, where $w_{i+1} = w_i a_i$. 

Thus, the play of a letter game can be seen as Adam constructing a word letter-by-letter, and Eve constructing a run transition-by-transition on the same word. Eve wins such a play if the following holds: if Adam's word  is in $L(\Ac)$, then Eve's run is accepting. 
\end{definition}
We note that the letter game is an $\omega$-regular game: the set of winning plays $\Pc$ for Eve are sequences of alternating letters and transitions, so that the word formed by just the letters is accepting in $\Ac$, while the run formed by just the transitions is rejecting. Since parity automata can be determinised, it is clear that $\Pc$ is an $\omega$-regular language, hence the letter game is an $\omega$-regular game, and therefore the letter game is determined~\cite{Mar75,GH82}.  

If Eve has a winning strategy on the letter game of $\Ac$, then $\Ac$ is said to be \emph{history-deterministic}. 
We are interested in the problem of checking whether a given parity automaton is history-deterministic, which we shall denote by \hdproblem.
\begin{quote}
    \hdproblem: Given a parity automaton $\Ac$, is $\Ac$ history-deterministic?
\end{quote}

\subsection{Token games}\label{sec:token-games}
Token games, or $k$-token games are defined on an automaton and are similar to letter games. Similar to as in a letter game, Adam constructs a word letter-by-letter and Eve constructs a run transition-by-transition on the same word over $\omega$ many rounds. But additionally, Adam also constructs $k$ runs transition-by-transition on that word. The winning objective of Eve requires her to construct an accepting run if one of $k$ Adam's runs is accepting.

\begin{definition}[$k$-token game]
    Given a nondeterministic parity automaton $\Ac = (Q,\Sigma,, q_0, \Delta,\Omega)$, the \emph{$k$-token game} of $\Ac$ is defined between the two players Adam and Eve as follows, with positions in $Q \times Q^k$. The game starts at $(q_0,(q_0)^k)$ and proceeds in $\omega$ many rounds. For each $i \in \mathbb{N}$, the $(i+1)^{th}$ round starts at a position $(q_i,(p^1_i,p^2_i,\cdots,p^k_i)) \in Q \times Q^k$, and proceeds as follows:\begin{itemize}
    \item Adam selects a letter $a_i \in \Sigma$
    \item Eve selects a transition $q_i \xrightarrow{a_i} q_{i+1} \in \Delta$
    \item Adam selects $k$ transitions $p^1_i \xrightarrow{a_i} p^1_{i+1}, p^2_i \xrightarrow{a_i} p^2_{i+1}, \cdots p^k_i \xrightarrow{a_i} p^k_{i+1},$
\end{itemize}
The new position is $(q_{i+1},(p^1_{i+1},p^2_{i+1},\cdots,p^k_{i+1}))$, from where the $(i+2)^{th}$ round begins.

Thus, in  a play of the $k$-token game, Eve constructs a run and Adam $k$ runs, all on the same word. Eve wins such a play if the following holds: if one of Adam's $k$ runs is accepting, then Eve's run is accepting.
\end{definition}

Bagnol and Kuperberg have shown that for any parity automaton $\Ac$, the $2$-token game of $\Ac$, and the $k$-token game of $\Ac$ for any $k \geq 2$, are equivalent.
\begin{lemma}[\cite{BK18}]
    Given a parity automaton $\Ac$, Eve wins $2$-token game of $\Ac$ if and only if Eve wins the $k$-token game of $\Ac$ for all $k \geq 2$. 
\end{lemma}

If $\Ac$ is a nondeterministic B\"uchi or co-B\"uchi automaton, then Eve wins the $2$-token game of $\Ac$ if and only if $\Ac$ is history-deterministic\cite{BK18,BKLS20b}, and it is conjectured that this result extends to all parity automata.

\begin{quote}
    {\scshape Two-token conjecture:} Given a nondeterministic parity automaton $\Ac$, Eve wins the $2$-token game of $\Ac$ if and only if $\Ac$ is history-deterministic.
\end{quote}

%% file: 3simulation.tex
In this section, we show that the problem of deciding if a parity automaton simulates another is $\NP$-hard, by giving a reduction from the problem of deciding whether Eve wins a 2-D parity game, which was shown to be $\NP$-complete by Chatterjee, Henzinger and Piterman~\cite{CHP07}. Since a simulation game can be solved in $\NP$ (see \cref{subsec:Prelims-simulation}), we obtain $\NP$-completeness. 
\begin{theorem}\label{thm:Sim-is-NP-hard}
Given two parity automata $\Ac$ and $\Bc$, deciding if $\Ac$ simulates $\Bc$ is $\NP$-complete.
\end{theorem}

Since $\Ac$ simulates $\Bc$ if and only if Eve wins the simulation game, which is a 2-dimensional parity game (see \cref{subsec:Prelims-simulation}), and deciding if Eve wins a 2-D parity game is in $\NP$~\cite{CHP07}, we get that the problem of checking for simulation is in $\NP$. Hence, we show that \simproblem is $\NP$-hard in the rest of this section, by giving a reduction from \twodimparity.

 Let $\Gc$ be a two-dimensional parity game played on the arena $G=(V,E)$, with two priority functions $\chi_1$ and $\chi_2$. We recall that the winning condition for Eve in such a game requires a play to satisfy the $\chi_2$-parity condition if the $\chi_1$-parity condition is satisfied (see \cref{def:twodimparitygame}).

\subsubsection*{Overview of the reduction.}
We shall construct two parity automata $\Hc$ and $\Dc$ such that $\Hc$ simulates $\Dc$ if and only if Eve wins $\Gc$. The automata $\Hc$ and $\Dc$ are over the alphabet $E \cup \{\$\}$, where $\$$ is a letter added for padding. The automaton $\Dc$ is deterministic, while the automaton $\Hc$ has nondeterminism on the letter $\$$ and contains a copy of $\Dc$. 

Adam, by his choice of letter in $Sim(\Hc,\Dc)$, captures his moves from Adam vertices in $\Gc$. Similarly, Eve, by means of choosing her transition on $\$$ in $\Hc$, captures her moves from Eve vertices in $\Gc$.  After each $\$$-round in $Sim(\Hc,\Dc)$, we require Adam to `replay' Eve's choice as the next letter. Otherwise, Eve can take a transition to the same state as Adam (recall that $\Hc$ contains a copy of $\Dc$), from where she wins the play in $Sim(\Hc,\Dc)$ by copying Adam's transitions in each round from here on-wards. The priorities of $\Dc$ are based on $\chi_1$, while the priorities of $\Hc$ are based on $\chi_2$. This way $\Dc$ and $\Hc$ roughly accept words that correspond to plays in $\Gc$ satisfying $\chi_1$ and $\chi_2$ respectively. 

We first present our reduction on an example 2-D parity game whose sub-game consists of vertices $u,v,v',w,w'$ with edges between them as shown in \cref{figure:sim-np-hardness}. For Adam's vertex $u$, we have corresponding states $u_D$ in $\Dc$ and $u_H$ in $\Hc$. An Adam move from $u$ in $\Gc$ corresponds to one round of $Sim(\Hc,\Dc)$ from the position $(u_D,u_H)$. 
In $\Gc$, Adam chooses an outgoing edge, say $e=(u,v)$ from $u$ such that $\chi_1(u)=c_1$ and $\chi_2(u)=c_2$. This corresponds to Adam choosing the letter $e$ in $Sim(\Hc,\Dc)$. We then have the corresponding unique transitions $u_D \xrightarrow{e:c_1} v_{\$}$ in $\Dc$ and $u_H \xrightarrow{e:c_2} v_H$ in $\Hc$, and hence the simulation game goes to $(v_{\$},v_H)$. 
\input{figure2simreduction}
An Eve move from $v$ in $\Gc$ corresponds to two rounds of the simulation game from $(v_{\$},v_H)$. In $Sim(\Hc,\Dc)$, Adam must select a letter $\$$ and the unique $\$$ transition $v_{\$} \xrightarrow{\$:0} v_D$ on $\Dc$, since $\$$ is the only letter on which there is an outgoing transition from $v_{\$}$. Eve must now select a transition on $\$$ from $v_{H}$. Suppose she picks $v_{H} \xrightarrow{\$:0} (v_{H},f)$ where $f=(v,w)$ is an outgoing edge from $v$ in $\Gc$ with $\chi_1(f)=c_5$ and $\chi_2(f)=c_6$. This corresponds to Eve selecting the edge $f$ from her vertex $v$ in $\Gc$. The simulation game goes to the position $(v_D,(v_H,f))$. From here, Adam may select any outgoing edge from $v$ as the letter. If he picks $f' = (v,w')$ and the transition $v_D \xrightarrow{f':c_7} w'_D$, then Eve can pick the transition $(v_{H},f) \xrightarrow{f':c_8} w'_D$ and move to the same state as Adam: such transitions are indicated by dashed edges in \cref{figure:sim-np-hardness}. From here, Eve can win $Sim(\Hc,\Dc)$ by simply copying Adam's transitions. Otherwise, Adam picks the edge $f$ as the letter, same as Eve's `choice' in the previous round, resulting in the transition $v_D \xrightarrow{f:c_5} w_D$ in $\Dc$ and $(v_H,f)\xrightarrow{f:c_6}w_H$ in $\Hc$, and the simulation game goes to the position $(w_D,w_H)$, from where the game continues similarly.

\subsubsection*{The reduction.}
We now give formal descriptions of the two parity automata $\Dc$ and $\Hc$ such that $\Hc$ simulates $\Dc$ if and only if Eve wins $\Gc$. We encourage the reader to refer to \cref{figure:sim-np-hardness} while reading the construction of the automata described below. 

Both automata $\Dc$ and $\Hc$ are over the alphabet $\Sigma = E \cup \{\$\}$. The automaton $\Dc$ is given by $\Dc = (P, \Sigma, p_0,\Delta_D, \Omega_D)$, where the set $P$ consists of the following states:
\begin{itemize}
    \item states $u_D$ for each Adam vertex $u \in V_{\adam}$,
    \item states $v_{\$}$ and $v_{D}$ for each Eve vertex $v \in V_{\eve}$.
\end{itemize}
The state $p_0 = \iota_D$ is the initial vertex, where $\iota$ is the initial vertex of the game $\Gc$.
The set $\Delta_D$ consists of the following transitions with their priorities (given by $\Omega_D$) as indicated:
\begin{itemize}
    \item transitions $u_D \xrightarrow{e:\chi_1(e)} v_{\$}$ for every edge $e = (u,v)$ in $\Gc$ such that $u \in V_{\adam}$ is an Adam-vertex in $\Gc$,
    \item transitions $v_{\$} \xrightarrow{\$:0} v_D$ for every $v \in V_{\eve}$ that is an Eve-vertex in $\Gc$,
    \item transitions $v_D \xrightarrow{f:\chi_1(f)} w_D$ for every edge $f=(v,w)$ in $\Gc$ such that $v \in V_{\eve}$ is an Eve vertex in $\Gc$
\end{itemize}

The automaton $\Hc$ is given by $\Hc = (Q, \Sigma, q_0, \Delta_H,\Omega_H)$, where the set $Q$ consists of the following states:
\begin{itemize}
    \item states $u_H$ for each Adam vertex $u \in V_{\adam}$,
    \item states $v_{H}$ for each Eve vertex $v \in V_{\eve}$,
    \item states $(v_{H},f)$ for each edge $f=(v,w)$ in $\Gc$ such that $v \in V_{\eve}$ is an Eve vertex,
    \item all states in $P$, the set of states of $\Dc$.
\end{itemize}
The state $q_0=\iota_H$ is the initial vertex. The set $\Delta_H$ consists of the following transitions with their priorities (given by $\Omega_H$) as indicated:
\begin{itemize}
    \item transitions $u_H \xrightarrow{e:\chi_2(e)} v_H$ for every edge $e = (u,v)$ in $\Gc$ such that $u \in V_{\adam}$ is an Adam-vertex in $\Gc$,
    \item transitions $v_H \xrightarrow{\$:0} (v_H,f)$ for every edge $f = (v,w)$ in $\Gc$ that is outgoing from an Eve-vertex $v \in V_{\eve}$,
    \item transitions $(v_H,f) \xrightarrow{f:\chi_2(f)} w_H$ for every edge $f=(v,w)$ in $\Gc$ outgoing from an Eve-vertex $v \in V_{\eve}$,
    \item transitions $(v_H,f) \xrightarrow{f':\chi_2(f')} w'_D$ for every edge $f'=(v,w') \neq f$ in $\Gc$ outgoing from an Eve-vertex $v \in V_{\eve}$,
    \item all transitions of $\Dc$.
\end{itemize}
Note that, by construction, $\Hc$ contains a copy of $\Dc$ as a sub-automaton. 

\subsubsection*{Correctness of the reduction.}
We now show that Eve wins the simulation game $Sim(\Hc,\Dc)$ if and only if Eve wins the game $\Gc$. Call any play of the simulation game \emph{\honest} if the following holds: whenever Eve's state in $\Hc$ is at $(v_H,f)$ at the start of a round of $Sim(\Hc,\Dc)$, Adam plays the letter $f$. If Adam plays a letter $f' \neq f$, then we call such a move \emph{\dishonest}. Any play consisting of a \dishonest move is called a \dishonest play.

It is clear that Eve wins any play of $Sim(\Hc,\Dc)$ that is \dishonest, since a \dishonest move causes Eve's state in $\Hc$ and Adam's state in $\Dc$ to be the same in $Sim(\Hc,\Dc)$. Then, both Eve's and Adam's runs are identical and determined by the choices of Adam's letters. In particular, Eve's run is accepting if Adam's run is. 

Thus, it suffices to consider only \honest plays. We first observe an invariant that is preserved throughout any \honest play of $Sim(\Hc,\Dc)$.
\begin{quote}\emph{Invariant:}
    At the start of any round of the simulation game $Sim(\Hc,\Dc)$ following an \honest play:
\begin{itemize}
    \item Adam's state is at $u_D$ for some $u \in V_{\adam}$ if and only if Eve's state is at $u_H$
    \item Adam's state is at $v_{\$}$ for some $v \in V_{\eve}$ if and only if Eve's state is at $v_H$
    \item Adam's state is at $v_D$ for some $v \in V_{\eve}$ if and only if Eve's state is at $(v_H,f)$ for some edge $f$ that is outgoing from $v$.
\end{itemize}
\end{quote}
This invariant is easy to observe from the construction, and can be shown by a routine inductive argument.

Note that if Adam constructs the word $w = e_0 \$ f_0 e_1 \$ f_1\ldots$---which we denote by $(e_i \$ f_i)_{i\geq 0}$ for succinctness---in an \honest play of $Sim(\Hc,\Dc)$, then Eve's run on $\Hc$ is uniquely determined, since the letter $f_i$ indicates how nondeterminism on $\Hc$ was resolved by Eve on the $i^{th}$ occurrence of $\$$ in $Sim(\Hc,\Dc)$. Thus, any \honest play in the simulation can be thought of as Adam selecting the $e_i$'s and Eve selecting the $f_i$'s, resulting in the word $w = (e_i \$ f_i)_{i \geq 0}$ being constructed in the simulation game. Note that then, by construction, $(e_i  f_i)_{i \geq 0}$ is a play in $\Gc$. Conversely, if $(e_i  f_i)_{i \geq 0}$ is a play in $\Gc$, then there is an \honest play of $Sim(\Hc,\Dc)$ whose word is $w = (e_i \$ f_i)_{i \geq 0}$. 

Furthermore, observe that the transitions on a letter $e \in E$ in $\Dc$ and $\Hc$ in any \honest play have the priorities $\chi_1(e)$ and $\chi_2(e)$ respectively, while transitions on $\$$ have priority 0. Thus, in a \honest play of $Sim(\Hc,\Dc)$ whose word is $(e_i \$ f_i)_{i \geq 0}$, the highest priorities occurring infinitely often in the run on $\Dc$ and $\Hc$ are the same as the highest $\chi_1$-priority and $\chi_2$-priority occurring infinitely often in the play $(e_i  f_i)_{i \geq 0}$ respectively. 

Thus, an \honest play in $Sim(\Hc,\Dc)$ whose word is $w = (e_i \$ f_i)_{i \geq 0}$ is winning for Eve if and only if the play $(e_i  f_i)_{i\geq 0}$ in $\Gc$ is winning for Eve. Since Eve wins any \dishonest play, the equivalence of the games $\Gc$ and $Sim(\Hc,\Dc)$ follows easily now. If Eve has a winning strategy in $\Gc$, she can use her strategy to select transitions so that the word $w=(e_i \$ f_i)_{i \geq 0}$ that is constructed in any \honest play $\rho$ of $Sim(\Hc,\Dc)$ corresponds to a winning play for her in $\Gc$, and hence $\rho$ is winning in $Sim(\Hc,\Dc)$. If Adam ever makes a \dishonest move, she wins trivially. 

Conversely, if she has a winning strategy in $Sim(\Hc,\Dc)$, then she can use her strategy to choose moves in $\Gc$ so that the play $(e_i  f_i)_{i\geq 0}$ corresponds to a winning \honest play of $Sim(\Hc,\Dc)$ in which the word $(e_i \$ f_i)_{i \geq 0}$ is constructed, thus resulting in the play $(e_i  f_i)_{i \geq 0}$ to also be winning for Eve.

%% file: figure2simreduction.tex
\begin{figure}
     \centering
    \begin{tikzpicture}
\usetikzlibrary {graphs,quotes}
    \node (G) at (-1,0) {$\Gc$:};
    \node (u) [regular polygon, regular polygon sides = 5, minimum size = 0.25 cm, draw]    at (0,0)   {$u$};
    \node (v) [rectangle, minimum size = 0.6 cm, draw] at (3,0) {$v$};
    \node (w) [regular polygon, regular polygon sides = 5, minimum size = 0.25 cm, draw] at (9,0) {$w$};
    \node (v') [rectangle, minimum size = 0.6 cm, draw] at (3,-2) {$v'$};
    \node (w') [regular polygon, regular polygon sides = 5, minimum size = 0.25 cm, draw] at (9,-2) {$w'$};
    \path[-stealth]
    (u) edge node [above] {$e:(c_1,c_2)$} (v)
    (u) edge node [above, rotate = 325] {$e':(c_3,c_4)$} (v')
    (v) edge node [above] {$f:(c_5,c_6)$} (w)
    (v) edge node [above, rotate = 342] {$f':(c_7,c_8)$} (w')
    (v') edge node [above] {$g:(c_9,c_{10})$} (w')
    ;

    \node (A) [Plum] at (-0.7,-4) {$\Dc$:};
    \node (ua) [circle, minimum size = 0.5cm, draw] at (0,-4) {$u_D$};
    \node (vdollar) [circle, minimum size = 0.5cm, draw] at (3,-4) {$v_{\$}$};
    \node (va) [circle, minimum size = 0.5cm, draw] at (6,-4) {$v_D$};
    \node (wa) [circle, minimum size = 0.5cm, draw] at (9,-4) {$w_D$};
    \node (vdollar') [circle, minimum size = 0.5cm, draw] at (3,-6) {$v'_{\$}$};
    \node (va') [circle, minimum size = 0.5cm, draw] at (6,-6) {$v'_D$}; 
    \node (wa') [circle, minimum size = 0.5cm, draw] at (9,-6) {$w'_D$};
    \path[-stealth]
    (ua) edge node [above] {$e:c_1$} (vdollar)
    (ua) edge node [above, rotate = 330] {$e':c_3$} (vdollar')
    (vdollar) edge node [above] {$\$:0$} (va)
    (vdollar') edge node [above] {$\$:0$} (va')
    (va) edge node [above] {$f:c_5$} (wa)
    (va) edge node [above, rotate = 330] {$f':c_7$} (wa')
    (va') edge node [above] {$g:c_9$} (wa')
    ;

    \node (B) [Dandelion] at (-1.3,-8) {$\Hc$:};
    \node (ub) [circle, minimum size = 0.5cm, draw] at (0,-8) {$u_H$};
    \node (vb) [circle, minimum size = 0.5cm, draw] at (3,-8) {$v_H$};
    \node (vb') [circle, minimum size = 0.5cm, draw] at (3,-11) {$v'_H$};
    \node (vbf) [ellipse, minimum size = 0.5cm, draw] at (6,-8) {$v_H,f$};
    \node (vbf') [ellipse, minimum size = 0.5cm, draw] at (6,-9.5) {$v_H,f'$};
    \node (vbg) [ellipse, minimum size = 0.5cm, draw] at (6,-11) {$v_H,g$};
    \node (wb) [circle, minimum size = 0.5cm, draw] at (9,-8) {$w_H$};
    \node (wb') [circle, minimum size = 0.5cm, draw] at (9,-11) {$w'_H$};
    \path[-stealth]
    (ub) edge node [above] {$e:c_2$} (vb)
    (ub) edge node [above, rotate = 315] {$e':c_4$} (vb')
    (vb) edge node [above] {$\$:0$} (vbf)
    (vb) edge node [above, rotate = 345] {$\$:0$} (vbf')
    (vb') edge node [above] {$\$:0$} (vbg)
    (vbf) edge node [above] {$f:c_6$} (wb)
    (vbf') edge node [above, rotate = 345] {$f':c_8$} (wb')
    (vbg) edge node [above] {$g:c_{10}$} (wb')
    (vbf) edge [red, dashed] node [above, rotate = 30] {$f':c_8$} (wa')
    (vbf') edge [red, dashed, bend right = 51] node [above, rotate = 40, near start] {$f:c_6$} (wa)
    ;

     \begin{scope}[on background layer]
    \node (s2) [draw = DarkOrchid, fill = Purple!20, inner sep = 1pt, rectangle, fit = (ua) (wa')] {};
    \end{scope}
    \begin{scope}[on background layer]
        \node (s1) [draw = Dandelion, fill = Dandelion!20, semitransparent, inner sep=1pt,rectangle,fit=(wb') (A) (s2)]  {};
    \end{scope}
    \end{tikzpicture}
      \caption{A snippet of a game $\Gc$, and the corresponding automata $\Dc$ and $\Hc$ constructed in the reduction. The Adam vertices are represented by pentagons and Eve vertices by squares. The automaton $\Dc$ is deterministic, and $\Hc$ contains a copy of $\Dc$.}
    \label{figure:sim-np-hardness}
\end{figure}

%% file: 4HDisNPhard.tex
In this section, we show that the problem of deciding whether a given nondeterministic parity automaton is history-deterministic is $\NP$-hard, as is the problem of deciding whether Eve wins the 1-token game or the 2-token game of a given parity automaton. To show this, we reduce from deciding whether Eve wins a 2-D parity game with priority functions $\chi_1$ and $\chi_2$ that satisfies the following property: any play satisfying the $\chi_2$-parity condition also satisfies the $\chi_1$-parity condition. We call such games `good 2-D parity games'. We first show in \cref{sec:good-games} that deciding whether Eve wins a \speshial 2-D parity game is $\NP$-hard, and then use this to show $\NP$-hardness for the problems mentioned above in \cref{sec:good-to-hd}.

\subsection{Good 2-D parity games}\label{sec:good-games}
\begin{definition}[Good 2-D parity game]
A 2-D parity game $\Gc$ with the priority functions $\chi_1$ and $\chi_2$ is called \emph{good} if any play in $\Gc$ that satisfies $\chi_2$ also satisfies $\chi_1$. 
\end{definition}
We call the problem of deciding whether Eve wins a 2-D parity game as {\scshape good 2-D parity game}.
Chatterjee, Henzinger and Piterman's reduction from SAT to \twodimparity~\cite{CHP07} can also be seen as a reduction to \specialparity, as we show below.

\begin{lemma}\label{lem:special-parity-is-np-hard}
Deciding whether Eve wins a \speshial 2-D parity game is $\NP$-hard.
\end{lemma}
\begin{proof}
We reduce from the problem of SAT. Let $\phi$ be a Boolean formula over the variables $X=\{x_1,x_2,\cdots,x_M\}$ that is a conjunction of terms $t_i$ for each $i \in [1,N]$, where each term $t_i$ is a finite disjunction of \emph{literals}---elements of the set $L=\{x_1,x_2,\cdots,x_M,\neg x_1,\neg x_2, \cdots, \neg x_M\}$. We shall construct a \speshial 2-D parity game $\Gc_{\phi}$ such that Eve wins $\Gc_{\phi}$ if and only if $\phi$ has a satisfying assignment. 

Let $T = \{t_1,t_2,\cdots,t_N\}$ be the set of all terms in $\phi$. The game $\Gc_{\phi}$ has the set $T \cup L$ as its set of vertices. The elements of $L$ are Adam vertices, while the elements of  $T$ are Eve vertices. We set the element $x_1$ in $L$ to be the initial vertex. Each Adam vertex $l$ in $L$ has an outgoing edge $e=(l,t)$ to every term $t$ in $T$, and every Eve vertex $t \in T$ has an outgoing edge $f=(t,l)$ to a literal $l$ if $l$ is a literal in $t$. Thus, each play in the game $\Gc_{\phi}$ can be seen as Adam and Eve choosing a term and a literal in that term in alternation respectively.

The game $\Gc_{\phi}$ has priority functions $\chi_1$ and $\chi_2$. To every edge $e = (l,t)$ that is outgoing from an Adam vertex, both priority functions $\chi_1$ and $\chi_2$ assign $e$ the priority 0, i.e., $\chi_1(e)=\chi_2(e)=0$. Every edge $e = (t,l)$ that is outgoing from an Eve vertex is assigned priorities as follows:
\begin{figure}[H]
    \centering
    \begin{minipage}{0.45\textwidth}
		\[\chi_1(e) = 
		\begin{cases}
			2j+2 & \text{ if } l = x_j \\
			2j+1 & \text{ if } l=\neg x_j \\
		\end{cases}\]
	\end{minipage}
    \begin{minipage}{0.45\textwidth}
		\[\chi_2(e) = 
		\begin{cases}
			2j & \text{ if } l = x_j \\
			2j+1 & \text{ if } l=\neg x_j \\
		\end{cases}\]
    \end{minipage}    
\end{figure}
This concludes our description of the game $\Gc_{\phi}$. We now show that $\Gc_{\phi}$ is a good 2-D parity game, which Eve wins if and only if $\phi$ is satisfiable. 

\subsubsection*{$\Gc_{\phi}$ is a good 2-D parity game.}
Let $\rho$ be a play in $\Gc_{\phi}$ that satisfies the $\chi_2$ parity condition. If $2c$ is the largest $\chi_2$-priority occurring infinitely often in $\rho$, then by construction, $2c+2$ is the largest $\chi_1$-priority occurring infinitely often in the $\rho$, which is also even. Thus, $\rho$ satisfies the $\chi_1$ parity condition. 

\subsubsection*{If $\phi$ is satisfiable, then Eve wins $\Gc_{\phi}$.}
Let $f:\{x_1,x_2,\cdots,x_M\} \xrightarrow{} \{\top,\bot\}$ be a satisfying assignment of $\phi$. Let $\sigma$ be a function which assigns, to each term $t_i$, a literal $l \in t_i$ that is assigned $\top$ in $f$. Consider the Eve-strategy $\sigma_{\eve}$ in $\Gc_{\phi}$ defined by $\sigma_{\eve}(t) = (t,\sigma(l))$. We claim that $\sigma_{\eve}$ is a winning strategy. Indeed, let $\rho$ be a play in $\Gc_{\phi}$ following $\sigma_{\eve}$, and consider the largest $i$ such that  $x_i$ or $\neg x_i$ appear infinitely often in $\rho$. Since $\sigma_{\eve}$ is obtained from a satisfying assignment, we know that either only $x_i$ appears infinitely often, or only $\neg x_i$ appears infinitely often. In the former case, the highest $\chi_2$ priority appearing infinitely often is $2i$, which is even, and hence $\rho$ is winning for Eve. In the latter case, the highest $\chi_1$ priority appearing infinitely often is $2i+1$, which is odd, and hence the $\chi_1$-parity condition is not satisfied, implying $\rho$ is winning for Eve.
\subsubsection*{If Eve wins $\Gc_{\phi}$, then $\phi$ is satisfiable.}
If Eve wins $\Gc_{\phi}$, then we know she can win using a positional strategy since $\Gc_{\phi}$ is a 2-dimensional parity game. Let $\sigma_{\eve}:T \xrightarrow{} L$ be such a strategy, where Eve chooses the edge $(t,\sigma_{\eve} (t))$ at a vertex $t$. If there are no two terms $t,t'$ such that $\sigma_{\eve} (t) =x_i$ and $\sigma_{\eve} (t') = \neg x_i$ for some $x_i$, then consider the assignment $\sigma$ defined as follows. The assignment $\sigma$ maps all variables $x$ that are in the image of $\sigma_{\eve}$ to $\top$, while any terms $x_j$ such that neither $x_j$ or $\neg x_j$ appear in the image are assigned $\top$ and $\bot$ respectively. It is clear then that $\sigma$ is a satisfying assignment, since each term $t$ in $\phi$ evaluates to $\top$. 

Otherwise, if there are terms $t,t'$ with $\sigma_{\eve}(t)=x_i$ and $\sigma_{\eve}(t')=\neg x_i$, we claim that Adam wins the game $\Gc_{\phi}$. Adam can alternate between picking $t$ and $t'$, and then the highest $\chi_1$ priority appearing infinitely often is $2i+2$ while the highest $\chi_2$ priority appearing infinitely often is $2i+1$. This implies that the play is winning for Adam, which is a contradiction since $\sigma_{\eve}$ is a winning strategy for Eve. \qed
\end{proof}
\subsection{$\NP$-hardness of checking history-determinism}\label{sec:good-to-hd}
We now  show that deciding the history-determinism, whether Eve wins the 1-token game, and whether Eve wins the 2-token game of a given parity automaton is $\NP$-hard (\cref{thm:HD-is-NP-hard}). Much of the work towards this has already been done in the reduction from \twodimparity to \simproblem given in \cref{sec:NP-hardnessforsimulation}. We show that the automaton $\Hc$ that is constructed when using this reduction from a \speshial 2-D parity game $\Gc$ is such that Eve wins $\Gc$ if and only if $\Hc$ is history-deterministic. Since \specialparity is $\NP$-hard (\cref{lem:special-parity-is-np-hard}), we get that \hdproblem is $\NP$-hard as well.

\begin{lemma}\label{lem:hd-is-np-hard}
    Checking whether a given nondeterministic parity automaton is history-deterministic is $\NP$-hard.
\end{lemma}
    \begin{proof}
    Let us consider a \speshial 2-D parity game $\Gc$. Recall the construction of the automata $\Hc$ and $\Dc$ in \cref{sec:NP-hardnessforsimulation}, which is such that Eve wins $\Gc$ if and only if $\Hc$ simulates $\Dc$. We will show that if $\Gc$ is a good 2-D parity game, then the following statements are equivalent.
\begin{enumerate}
    \item Eve wins $\Gc$.
    \item $\Hc$ simulates $\Dc$.
    \item $\Hc$ is history-deterministic.
\end{enumerate}
The equivalence of 1 and 3 would then conclude the proof. The equivalence of 1 and 2 has already been shown in the proof of \cref{thm:Sim-is-NP-hard}, and we now focus on showing that 2 and 3 are equivalent. 

Towards this, let $\Sigma = E \cup \{\$\}$, and consider the languages $\Lc_j$  over $\Sigma$ consisting of the words $(e_i  \$  f_i)_{i \geq 0}$ such that $(e_i f_i)_{i\geq 0}$ is a play in $\Gc$ that satisfies $\chi_j$, for $j = 1,2$. By construction, we know $L(\Dc) = \Lc_1$, and $L(\Hc) = \Lc_1 \cup \Lc_2$. Furthermore, since $\Gc$ is \speshial, we know that $\Lc_1 \supseteq \Lc_2$ and hence $L(\Dc) = L(\Hc)$. Observe that by construction, $\Dc$ is deterministic. 

If $\Hc$ is history-deterministic, then since $L(\Dc) = L(\Hc)$, Eve wins the simulation game between $\Hc$ and $\Dc$: she can use her strategy in the letter game of $\Hc$ to play in $Sim(\Hc,\Dc)$, ignoring Adam's transitions in $\Dc$. 

The converse direction follows from \cite[Theorem 4.1]{HP06}, where Henzinger and Piterman show that if a nondeterministic parity automaton $\Nc$ simulates a language-equivalent deterministic parity automaton, then $\Nc$ is history-determin\-istic. We include a proof nevertheless, for self-containment. Supposing $\Hc$ simulates $\Dc$, Eve can use her winning strategy in $Sim(\Hc,\Dc)$ to win the letter game of $\Hc$ as follows. Eve, during the letter game of $\Hc$, will keep in her memory, a play of the game $Sim(\Hc,\Dc)$. On each round in the letter game of $\Hc$, Adam gives a letter, and Eve, in the game $Sim (\Hc,\Dc)$, lets Adam pick the same letter and the unique transition on that letter in $\Dc$. She then uses her strategy in $Sim(\Hc,\Dc)$ to pick a transition in $\Hc$, and she plays the same transition in the letter game of $\Hc$. We claim that any resulting play of the letter game of $\Hc$ if Eve plays as above is winning for Eve. Indeed, if Adam constructs an accepting word in $\Hc$, then it is accepting in $\Dc$ as well. Hence, since $\Dc$ is deterministic, Adam's run on $\Dc$ in the simulation game between $\Hc$ and $\Dc$ that is stored in Eve's memory is accepting. Since Eve is playing according to a winning strategy in $Sim(\Hc,\Dc)$, Eve's run in $\Hc$, which is the same in $Sim(\Hc,\Dc)$ and the letter game of $\Hc$, is accepting as well. Hence, Eve wins the letter game of $\Hc$, and thus $\Hc$ is history-deterministic. \qed
\end{proof}

We also argue in the  appendix that the automaton $\Hc$ in proof of \cref{lem:hd-is-np-hard} above is such that Eve wins the 1-token game of $\Hc$ \emph{if and only if} Eve wins the 2-token game of $\Hc$ \emph{if and only if} $\Hc$ is history-deterministic.  This gives us that checking whether Eve wins the 1-token game or the 2-token game of a parity automaton is $\NP$-hard. Since $1$-token games can naturally be seen as a 2-D parity game, we get that deciding whether Eve wins the 1-token game of a given parity automaton is in $\NP$, and hence the problem is $\NP$-complete.

\begin{theorem}\label{thm:HD-is-NP-hard}
The following problems are $\NP$-hard:
\begin{enumerate}
    \item Given a parity automaton $\Ac$, is $\Ac$ history-deterministic?
    \item Given a parity automaton $\Ac$, does Eve win the $2$-token game of $\Ac$?
\end{enumerate}
Additionally, the following problem is $\NP$-complete: Given a parity automaton $\Ac$, does Eve win the $1$-token game of $\Ac$?
\end{theorem}

%% file: 5Inclusion.tex
In this section, we consider the following problem:
\begin{quote}
    {\scshape HD-automaton containment}: Given two parity automata $\Ac$ and $\Bc$ such that $\Bc$ is history-deterministic, is $L(\Ac) \subseteq L(\Bc)$?
\end{quote}
While the problem of checking language inclusion between two non-deterministic parity automata is $\PSPACE$-complete (regardless of whether the parity index is fixed or not)~\cite{KV98,ACCHHMV11}, the same for deterministic parity automata is $\NL$-complete~\cite[Theorem 1]{Sch10}. For history-deterministic parity automata with fixed parity indices, however, the problem of language inclusion reduces to checking for simulation (\cref{lem:inclusion-to-sim}), which can be solved in polynomial time when the parity indices are fixed~\cite{CHP07}. This gives us that checking for language inclusion between two history-deterministic parity automata with fixed parity index can be done in polynomial time (\cref{cor:fixed-parity-index-inc-in-polytime}). This observation has been treated as folklore, and we prove it here for completeness.   

\begin{lemma}[\cite{Sch20,BHLST23}]\label{lem:inclusion-to-sim}
Given a nondeterministic parity automaton $\Ac$ and a history-deterministic parity automaton $\Bc$, the following are equivalent: 
\begin{enumerate}
    \item $\Bc$ simulates $\Ac$
    \item $L(\Ac) \subseteq L(\Bc)$
\end{enumerate}
\end{lemma}
\begin{proof}
(1) $\Rightarrow$ (2): Fix $\sigma_{\eve}$ to be a winning strategy for Eve in $Sim(\Bc,\Ac)$. Let $w$ be a word  accepted by $\Ac$ via an accepting run $\rho$. Consider a play of $Sim(\Bc,\Ac)$ where Adam constructs the run $\rho$ on the word $w$, and Eve plays according to $\sigma_{\eve}$. Then, the run in $\Bc$ that Eve constructs must be accepting, and hence $w$ is accepted by $\Bc$.

(2) $\Rightarrow$ (1): Let $\sigma_{B}$ be a winning strategy for Eve in the letter game of $\Bc$. Consider the strategy for Eve in $Sim(\Bc,\Ac)$ where Eve chooses the transitions on $\Bc$ according to $\sigma_{B}$, ignoring Adam's transitions in $\Ac$. If Adam constructs an accepting run in $\Ac$ on a word $w$ in $Sim(\Bc,\Ac)$, then  $w \in L(\Ac) \subseteq L(\Bc)$. Hence $\sigma_{B}$ would have constructed an accepting run in $\Bc$ in $Sim(\Bc,\Ac)$. It follows that Eve wins $Sim(\Bc,\Ac)$, and hence $\Bc$ simulates $\Ac$. \qed
\end{proof}

\begin{corollary}\label{cor:fixed-parity-index-inc-in-polytime}
Given a nondeterministic parity automaton $\Ac$ and a history-deterministic parity automaton $\Bc$ such that both $\Ac$ and $\Bc$ have priorities in $[d]$ for a fixed $d$, the problem of whether $L(\Ac) \subseteq L(\Bc)$ can be decided in polynomial time.  
\end{corollary}

We now focus on the problem \hdcontainment when the parity index is not fixed. From \cref{lem:inclusion-to-sim}, we know that this can be reduced to \simproblem. Since \simproblem is in $\NP$~\cite{CHP07}, we get an immediate $\NP$-upper bound for \hdcontainment\cite[Lemma 3]{Sch20}. We show that we can do better, in quasi-polynomial time, by giving a polynomial time reduction to finding the winner in a parity game\cite{CaludeJKLS22,JL17}. 

Towards this, let us fix a nondeterministic parity automaton $\Ac$ and a history-deterministic parity automaton $\Bc$ over the alphabet $\Sigma$ throughout the rest of this section, for which we want to decide if $L(\Ac) \subseteq L(\Bc)$. Suppose that $\Ac$ has $n_1$ states and priorities in $[d_1]$, and $\Bc$ has $n_2$ states and priorities in $[d_2]$. 

It is well known that every such parity automaton $\Ac$ can be converted efficiently to a language-equivalent nondeterministic B\"uchi automaton $\Ac'$ that has at most $(n_1 \cdot d_1)$ states~\cite{Cho74,SN99}. Then, from \cref{lem:inclusion-to-sim}, it suffices to check if Eve wins the game $Sim(\Bc,\Ac')$. Note that $Sim(\Bc,\Ac')$ is a 2-D parity game $\Gc$ with $(n_1 \cdot d_1 \cdot n_2 \cdot |\Sigma|)$-many vertices that has the priority functions $\chi_1: V \xrightarrow{} [1,2]$ and $\chi_2:V \xrightarrow{} [d_2]$, where $V$ is the set of vertices of $\Gc$. 

The game $\Gc$ can be viewed equivalently as a Muller game with the condition $(C,\Fc)$, where $C = [1,2] \times [d_2]$ and $\Fc$ consists of sets $F \subseteq C$ such that if $\max (F\arrowvert_1)$ is even, then $\max (F\arrowvert_2)$ is even. Here, $F\arrowvert_i$ for $i\in \{1,2\}$ indicates the projection of $F$ onto the $i^{th}$ component. Call the Zielonka tree (\cref{def:Zielonkatree}) of this Muller condition as $Z_{d_2}$. We shall show that the size of $Z_{d_2}$ is polynomial in $d_2$.

\begin{restatable}{lemma}{appendixtwo}\label{lem:Appendix2}
 The Zielonka tree $Z_{d_2}$ has $(\lceil \frac{d_2}{2} \rceil)$ many leaves and its height is $d_2$.
\end{restatable}
The proof of the lemma, obtained via an inductive argument, can be found in the appendix. \cref{lem:Appendix2} allows us to use \cref{lem:Muller to Parity reduction} on $Sim(\Bc,\Ac')$ to obtain an equivalent Parity game $\Gc'$ with $(n_1 \cdot d_1 \cdot n_2 \cdot |\Sigma| \cdot \lceil \frac{d_2}{2} \rceil)$ vertices which has $d_2+1$ priorities, such that Eve wins $Sim(\Bc,\Ac')$ if and only if Eve wins $\Gc'$. 

\begin{lemma}\label{lem:lastlemma}
    Given a nondeterministic parity automaton $\Ac$ with $n_1$ states and a history-deterministic parity automaton $\Bc$ with $n_2$ states whose priorities are in $[d_2]$ that are both over the alphabet $\Sigma$, the problem of deciding whether $L(\Ac)$ is contained in $L(\Bc)$ can be reduced in polynomial time to finding the winner of a parity game $\Gc$ which has $(n_1 \cdot d_1 \cdot n_2 \cdot |\Sigma| \cdot \lceil \frac{d_2}{2} \rceil)$ many vertices and $d_2+1$ priorities. 
\end{lemma}
Since parity games can be solved in quasi-polynomial time\cite{CaludeJKLS22,JL17}, \cref{lem:lastlemma} implies that the problem of language containment in a history-deterministic automaton can be solved in quasi-polynomial time as well.

\begin{theorem}\label{thm:inclusion-theorem}
Given a nondeterministic parity automaton $\Ac$ with $n_1$ states and priorities in $[d_1]$, and a history-deterministic parity automaton $\Bc$ with $n_2$ states whose priorities are in $[d_2]$, checking whether the language of $\Ac$ is contained in the language of $\Bc$ can be done in time $$(n_1 \cdot d_1 \cdot n_2 \cdot d_2 \cdot |\Sigma|)^{\Oc({\log d_2})}.$$ 
\end{theorem}

%% file: 6conclusion.tex
We have shown $\NP$-hardness for the problem of checking for simulation between two parity automata (when their parity indices are not fixed). We have also established upper and lower bounds of several decision problems relating to history-deterministic parity automata. The most significant amongst these, in our view, is the $\NP$-hardness for the problem of deciding if a given parity automaton is history-deterministic, which is an improvement from the previous lower bound of solving a parity game~\cite{KS15}.
 
There still remains a significant gap between the lower bound of $\NP$-hardness and the upper bound of $\EXPTIME$ for checking history-determinism, however. Furthermore, note that even if one shows the two-token conjecture~\cite{BK18,BKLS20b}, this would only imply a $\PSPACE$-upper bound (when the parity index is not fixed), since 2-token games can be seen as Emerson-Lei games~\cite{HD05}. Thus, a natural direction for future research is to try to show that the problem of checking for history-determinism is $\PSPACE$-hard.

On the other hand, however, it is also plausible that checking whether Eve wins the 2-token game of a given parity automaton can be done in $\NP$. A proof for this might show that if Eve wins a 2-token game, then she has a strategy that can be represented and verified polynomially. Such an approach, which would involve understanding the strategies for the players in the 2-token games better, could also yield crucial insights for proving or disproving the two-token conjecture (see \cref{sec:token-games}). 

Boker and Lehtinen showed in their recent survey that for a `natural' class of automata $T$, checking history-determinism for $T$-automata is at least as hard as solving $T$-games~\cite{BL23}. Interestingly, the problem of checking history-determinism over $T$-automata also has the matching upper bound of solving $T$-games for all classes of automata $T$ over finite words, and over infinite words with safety and reachability objectives on which the notion of history-determinism has been studied so far~\cite{BL22,GJLZ21,PT23,BHLST23,EGJLZ23}.
Our result of the problem of checking history-determinism being $\NP$-hard for parity automata deviates from this trend (unless parity games are $\NP$-hard, which would have the drastic and unlikely consequence of $\NP = \NP \cap \coNP$), and demonstrates the additional intricacy that parity conditions bring.

\subsubsection{Acknowledgements}
We thank Marcin Jurdzi\'nski, Neha Rino, K. S. Thejaswini, and anonymous reviewers for their feedback and suggesting numerous improvements to the paper. Additionally, we are grateful to K. S. Thejaswini for several insightful discussions and pointing out a flaw in an earlier proof of \cref{thm:Sim-is-NP-hard}. 

%% file: appendix1.tex
\begin{lemma}\label{lemma:token-games-are-hard}
The following problems are $\NP$-hard:
\begin{enumerate}
    \item Given a parity automaton $\Ac$, does Eve win the $1$-token game of $\Ac$?
    \item Given a parity automaton $\Ac$, does Eve win the $2$-token game of $\Ac$?
\end{enumerate}
\end{lemma}
\begin{proof}
 As in the proof of \cref{lem:hd-is-np-hard}, consider a \speshial 2-D parity game $\Gc$, and let $\Hc$ and $\Dc$ be automaton constructed in \cref{sec:NP-hardnessforsimulation}, which is such that Eve wins $\Gc$ if and only if $\Hc$ simulates $\Ac$. From \cref{lem:hd-is-np-hard}, we also know that Eve wins $\Gc$ if and only if $\Hc$ is HD. We will show that the following holds: 
 \begin{enumerate}
     \item if $\Hc$ is history-deterministic, then Eve wins the 1-token game and the 2-token game of $\Hc$,
     \item if $\Hc$ is not history-deterministic, then Adam wins the 1-token game and the 2-token game of $\Hc$.
 \end{enumerate}
(1) is easy to see, since if $\Hc$ is history-deterministic, then Eve can use her letter game strategy to pick transitions on her token in the 1-token game or the 2-token game, ignoring Adam's tokens.

For (2), we show that Adam wins the 1-token game of $\Hc$ if $\Hc$ is not HD. This would imply that Adam wins the 2-token game of $\Hc$ as well, since he can use his strategy in 1-token game to win the 2-token game as follows: he copies Eve's transitions in his second token and follows his strategy in the 1-token game of $\Hc$ to choose the letters and transitions on his first token. 

Since Adam wins the letter game of $\Hc$, we know that $\Hc$ does not simulate $\Dc$. Fix a winning strategy $\sigma_{\adam}$ of Adam in the simulation game. We now describe how Adam can win the 1-token game of $\Hc$ by using $\sigma_{\adam}$. At a high level, Adam will exploit the nondeterminism on $\$$ to ensure his token eventually moves to $\Dc$ and traces out an accepting run, while picking letters according to $\sigma_{\adam}$ and ensuring Eve's run on her token is rejecting.

In more details, Adam in the 1-token game will keep a play of the game $Sim(\Hc,\Dc)$ in his memory in order to pick the letters and transitions in the 1-token game of $\Hc$. He picks letters in the 1-token game of $\Hc$ using $\sigma_{\adam}$, by viewing Eve's transitions in the 1-token game as transitions on $\Hc$ in $Sim(\Hc,\Dc)$. Note that since $\Dc$ is deterministic, Adam's transitions in $\Dc$ in $Sim(\Hc,\Dc)$ depends solely on Adam's choice of letters. For Adam's token in the 1-token game, he copies Eve's transitions till Eve hasn't had to resolve nondeterminism, i.e., there is a unique transition from Eve's state on Adam's letter in the corresponding round. Eventually, however, there must be a round where Eve would need to resolve nondeterminism: otherwise, since $L(\Dc) = L(\Hc)$, Eve would construct an accepting run if the word is accepting, and in particular, if Adam's run on his token is accepting.

Thus, let Eve's token be at the state $q$ when she needed to resolve the nondeterminism. Then, Adam's token is also at $q$, and by construction, we know that $q=v_H$ for some $v \in V_{\eve}$. In the corresponding simulation game $Sim(\Hc,\Dc)$ in Adam's memory, his state is $v_{\$}$. Now, Adam gives the letter $\$$, and his state in the simulation game is then updated to $v_D$, while Eve picks a transition $v_H \xrightarrow{\$:0} (v_H,f)$ where $f=(v,w)$ is an outgoing edge from $v$. Adam will then pick a transition $v_H \xrightarrow{\$:0} (v_H,f')$ for some edge $f' \neq f$ on his token. 

In the next round, Adam must pick the letter $f$ according to $\sigma_{\adam}$, or Adam's state and Eve's state in $Sim(\Hc,\Dc)$ would both be the same and in $\Dc$, causing Eve to win, which is a contradiction since  $\sigma_{\adam}$ is a winning strategy. This causes Adam's state in $\Dc$ in $Sim(\Hc,\Dc)$ to be $w_D$. In the $1$-token game, Eve's token goes to $w_H$ and Adam's token goes to $w_D$---same as his state in $Sim(\Hc,\Dc)$. From here, Adam can continue choosing letters according to $\sigma_{\adam}$, while his transitions are uniquely determined. Since $\sigma_{\adam}$ is a winning strategy, Eve's run on her token is rejecting, while Adam's run $\rho$ in $\Dc$ in $Sim(\Hc,\Dc)$ is accepting. Since the run of Adam's token eventually coincides with $\rho$, it is accepting as well, and hence Adam wins the $1$-token game.  \qed
\end{proof}

%% file: appendix2.tex
\appendixtwo*
\begin{proof}
    We prove the lemma by induction on $d_2$. When $d_2$ is $0$ or $1$, the tree $Z_{d_2}$ is as shown below in \cref{figure:zielona-tree-base-case}, and the induction hypothesis is clearly satisfied.
\begin{figure}[H]
    \centering
    \begin{tikzpicture}
        \node (a) [blue] at (0,-1.5) {$\{1,2\} \times [0]$};      
        \node (c) [red] at (5,0) {$\{1,2\} \times [1]$};
        \node (d) [blue] at (3.8,-1.5) {$\{1,2\} \times [0]$};
        \node (e) [blue] at (6.2,-1.5) {$\{1\} \times [1]$};
        \graph{(c) -> {(d),(e)}};
    \end{tikzpicture}
      \caption{Zielonka trees $Z_0$ and $Z_1$}
    \label{figure:zielona-tree-base-case}
\end{figure}
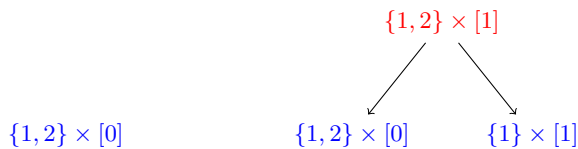

Let $d_2 \geq 2$, and suppose that the lemma holds for all $d' < d_2$. We distinguish between the cases of when $d_2$ is even or odd.

If $d_2 = 2k$ for some $k>0$, then the root of the node is labelled by the set $[1,2] \times [2k]$, which is in $\Fc$. Then, the only maximal set which is a proper subset of $[1,2] \times [2k]$ and not in $\Fc$ is $[1,2] \times [2k-1]$. Thus, the child $c$ of the root is labelled $[1,2] \times [2k-1]$ and $c$ is then the root of the tree $Z_{2k-1}$ itself (see \cref{figure:zielona-tree-induction-proof}). By induction hypothesis, the tree $Z_{2k-1}$ has height $2k-1$ and $k$ leaves, and hence, the tree $Z_{2k}$ has height $2k-1+1=2k$ and $k$ leaves, as desired.
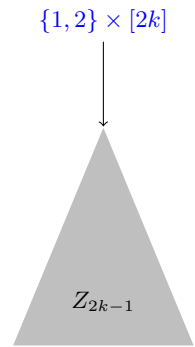
\begin{figure}
    \centering
    \begin{tikzpicture}
        \node (a) [blue] at (0,0) {$\{1,2\} \times [2k]$};
        \node  (b)  [fill=gray!50,inner sep=1pt,isosceles triangle, anchor = apex, shape border rotate = 90, minimum height = 1.5 cm, minimum width = 2.4 cm] at (0,-1.4) {$Z_{2k-1}$};
%        \node (c) [red] at (6,0) {$\{1,2\} \times [2k+1]$};
%        \node  (d)  [fill=gray!50,inner sep=1pt,isosceles triangle, anchor = apex, shape border rotate = 90, minimum height = 1.5 cm, minimum width = 2.4 cm] at (4.9,-1.4) {$Z_{2k}$};
%        \node (d') at (4.8,-1.5) {};
%        \node (e) [blue] at (7.2,-1.5) {$\{1\} \times [2k+1]$};
        \graph{(a)->(b)};
    \end{tikzpicture}
      \caption{Zielonka tree  $Z_{2k}$}
    \label{figure:zielona-tree-induction-proof}
\end{figure}

If $d_2 = 2k+1$ for some $k>0$, then the root of the node is labelled by the set $[1,2] \times [2k+1]$, which is not in $\Fc$. Now, there are two maximal sets that are proper subsets of $[1,2] \times [2k+1]$ and are in $\Fc$: the set $\{1\} \times [2k+1]$ which then has no proper subsets that are not in $\Fc$, and the set $\{1,2\} \times [2k]$. Accordingly, the root of $Z_{2k+1}$ has two children, one labelled by $\{1,2\} \times [2k]$ that is then the the root of the tree $Z_{2k}$, and another labelled $\{1\} \times [2k+1]$, as shown in \cref{figure:zielona-tree-induction-proof-Z2k+1}. By induction hypothesis, the tree $Z_{2k}$ has height $2k$ and $k$ leaves, and hence, the tree $Z_{2k+1}$ has height $2k+1$ and $k+1$ leaves, as desired.
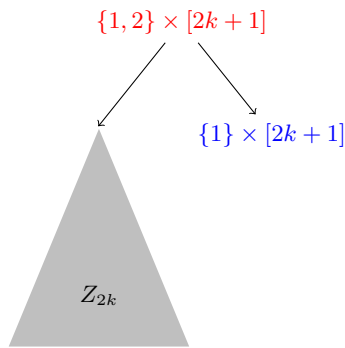
\begin{figure}
    \centering
    \begin{tikzpicture}
        \node (c) [red] at (6,0) {$\{1,2\} \times [2k+1]$};
        \node  (d)  [fill=gray!50,inner sep=1pt,isosceles triangle, anchor = apex, shape border rotate = 90, minimum height = 1.5 cm, minimum width = 2.4 cm] at (4.9,-1.4) {$Z_{2k}$};
        \node (d') at (4.8,-1.5) {};
        \node (e) [blue] at (7.2,-1.5) {$\{1\} \times [2k+1]$};
        \graph{(c) -> {(d'),(e)}};
    \end{tikzpicture}
      \caption{Zielonka tree $Z_{2k+1}$}
    \label{figure:zielona-tree-induction-proof-Z2k+1}
\end{figure}
\qed
\end{proof}